\documentclass{aamas2016}
\begin{document}

\title{On the Power of Dominated Players in Team Competitions
\titlenote{This work was supported by the National Basic Research Program of China Grant 2011CBA00300, 2011CBA00301, the Natural Science Foundation
of China Grant 61033001, 61361136003, 61303077, 61561146398, a Tsinghua Initiative Scientific Research Grant and a China Youth 1000-talent program.}}

\numberofauthors{3}

\author{
\alignauthor
Kai Jin\\
       \affaddr{Institute for Interdisciplinary Information Sciences}\\
       \affaddr{Tsinghua University}\\
       \affaddr{Beijing, China}\\
       \email{cscjjk@gmail.com}
\alignauthor
Pingzhong Tang\\
       \affaddr{Institute for Interdisciplinary Information Sciences}\\
       \affaddr{Tsinghua University}\\
       \affaddr{Beijing, China}\\
       \email{kenshinping@gmail.com}
\alignauthor
Shiteng Chen\\
       \affaddr{Institute for Interdisciplinary Information Sciences}\\
       \affaddr{Tsinghua University}\\
       \affaddr{Beijing, China}\\
       \email{cst100@gmail.com}
}

\maketitle

\begin{abstract}
We investigate multi-round team competitions between two teams, where each team selects one of its players simultaneously in each round and each player can play at most once. The competition defines an extensive-form game with perfect recall and can be solved efficiently by standard methods. We are interested in the properties of the subgame perfect equilibria of this game.

We first show that uniformly random strategy is a subgame perfect equilibrium strategy for both teams when there are no redundant players (i.e., the number of players in each team equals to the number of rounds of the competition). Secondly, a team can safely abandon its weak players if it has redundant players and the strength of players is transitive.

We then focus on the more interesting case where there are redundant players and the strength of players is not transitive.
In this case, we obtain several counterintuitive results. First of all, a player might help improve the payoff of its team, even if it is dominated by the entire other team.
We give a necessary condition for a dominated player to be useful. We also study the extent to which the dominated players can increase the payoff.

These results bring insights into playing and designing general team competitions.
\end{abstract}


\terms{Algorithms, Economics, Theory}


\newtheorem{theorem}{Theorem}
\newtheorem{lemma}{Lemma}
\newtheorem{claim}{Claim}
\newtheorem{corollary}{Corollary}
\newtheorem{conjecture}{Conjecture}
\newdef{definition}{Definition}
\newdef{example}{Example}
\newdef{remark}{Remark}
\newdef{question}{Question}

\section{Introduction}

We investigate a type of {\em team competitions} where there are two teams, each with a number of players, competing against each other. The competition proceeds in a fixed number of rounds. In each round, each team simultaneously sends out a player to a match. The result of the match is then revealed according to a probabilistic strength matrix between players. The selected players cannot compete in the subsequent rounds. The competition proceeds to the next round if there is one; or terminated otherwise. The format of the competition and the strength matrix are common knowledge to both teams. The final payoff of each team is the number of matches it wins. We also consider another commonly seen form where each team gets payoff $1$ if it wins strictly more matches than the other team, $0$ if ties, and $-1$ if it wins less matches. Clearly, the competition defines a standard {\em extensive-form game}, or more precisely, a \emph{stacked matrix game} (\cite{lanctot:Goof}). We are interested in the sub-game perfect equilibria of the game, i.e., a strategy profile that specifies for each team which player to play at each round.

We are particularly interested in a situation where at least one team has more players than the number of rounds in the competition. As a result, some players will never have chance to participate in any match. The main agenda of this paper is to understand to what extent can the presence of additional players affect the payoff of both teams. In particular, we ask the following questions:
\begin{enumerate}
\item Can the presence of additional weakest teammate, a teammate whose row in the strength matrix is strictly dominated by any other row, help increase the payoff of the team?
\item Can the presence of additional dominated teammate, a teammate that always loses to any player in the opponent team, help increase the payoff of the team?
\end{enumerate}

It might appear intuitive that the answers to both questions are negative. For the first question, it seems that the weakest teammate will never have a chance to participate in any match since one can always replace him by a better teammate and increase payoff. For the second question, it might seem more obvious since the dominated teammate will lose any matches thus must be replaced by a better teammate. To our surprise, we find that the answers to both questions are affirmative.

\subsection{Our contributions}

We first show that uniformly random strategy is a sub-
game perfect equilibrium strategy for both teams when there
are no redundant players (i.e., the number of players in each
team equals to the number of rounds of the competition).
The uniformly random strategy always picks the unmatched player uniformly at random in each round.

Then, we consider the general case where at least one team has redundant players.
We first study the case where the strength of players is transitive,
which means that the players can be rearranged in a queue so that each of them is weaker than its successor.
We prove that, a team can safely abandon its weak players if it has
redundant players and the strength of players is transitive.
Therefore, this case reduces to the case where there are no redundant players.

Finally, we focus on the case where there are redundant players and the strength of players is not transitive.
In this case, we obtain a number of counterintuitive results.
Most importantly, a player might help improve its team's payoff,
even if it is dominated by the entire opposing team.
We give a necessary condition for a dominated player to be useful, which alternatively suggest that a particular utility function (named $U_E$ below) is more reasonable in team competition.
Our results imply that a team can increase its utility by recruiting additional dominated players.
We further show that, the optimal number of dominated players to recruit can scale with the number of rounds.
More precisely, this number can be $\Theta(T)$ if there are $T$ rounds.
Last but not least, we study the limitation of dominated players.

\subsection{Motivations}

Our model is first motivated by the Chinese horse race story described in \cite{Tang09} (see also \cite{wiki_horse}).
It represents one of the most popular forms of horse races where each team ranks its horses to match sequentially.
Moreover, the Swaythling Cup, as known as World Table Tennis Championships, follows the same model described in our paper: each team adaptively selects a ranking of three players and brings two additional substitutes. In fact, this has been one of the most popular formats of team competition in table tennis.
In addition, the card game Goofspeil (\cite{lanctot:Goof,wiki_goofspeil}) also falls into nearly the same model as described in our paper.

\subsection{Related works}

Tang, Shoham and Lin~\cite{Tang09,Tang10} study a team competition setting where the number of players equals the number of rounds and both teams must determine the ordering of players upfront, before the competition starts. They put forward competition rules that are truthful while satisfy other desirable properties. The main difference between their work and ours is that we do not design new mechanisms but study game theoretical properties of commonly used competition rules. The differences also lie in that the strategies are adaptive in our setting and each team can have more players than the number of matches.

The strategic aspects of team competition have also been under scrutiny of computer scientists due to a recent Olympic scandal in badminton, where several teams deliberately throw matches in order to avoid a strong opponent in the next round. The phenomenon has been discussed in depth in a series of algorithmic game theory blogposts by Kleinberg~\cite{Kleinberg12} and Procaccia~\cite{Procaccia13}.

A parallel literature has been concerned with the strategic aspect of tournament seeding~\cite{Hwang1982,Rosen85,Knuth87,Schwenk2000,Vu09,Altman09}. It is well known that there are cases where by strategic seeding and structuring, any player can be winner in knockout tournament. Various game theoretical questions, such as player-optimal seeding, complexity of manipulation and incentives to guarantee strategyproofness, have been investigated in this literature.

\section{The model}

Team~1 has a set of $m$ players $\{A_1,\ldots, A_m\}$. Team~2 has a set of $n$ players $\{B_1,\ldots, B_n\}$.
A competition between team 1 and 2 is a tuple $G(T,P,U)$ where,
\begin{enumerate}
\item $T$ is the number of rounds.
\item In each round, each team simultaneously selects one of its players that have not been selected yet.
\item $P$ is a probabilistic matrix that describes the relative strength between players, with $P_{i,j}$ denoting the probability that $A_i$ wins against $B_j$ and $1-P_{i,j}$ the probability for $A_i$ to lose to $B_j$.
\item $U: [T]\rightarrow R$ denotes the utility function of each team. The utility function only depends on the number of rounds $t$ each team wins, i.e., it can be represented by $U(t)$. This also implies that both teams have the same utility functions.
\item The parameters $n, m, T, P, U$ are common knowledge to both teams, and historical plays are perfectly observable. It is assured that $P_{i,j}\in [0,1]$ for all $i,j$
and that $m\geq T$ and $n\geq T$ so that there are enough players to complete the competition.
\end{enumerate}

The following utility functions $U_E$ and $U_M$ are two commonly seen ones:

\begin{equation}
U_E(t)=t-T/2. \quad U_M(t)=\left\{\begin{array}{cc}
                                    1 & t>T/2 \\
                                    0 & t=T/2 \\
                                    -1 & t<T/2
                                  \end{array}
\right..
\end{equation}

In other words, $U_E$ describes a competition where a team's utility is exactly the number of rounds it wins (minus some constant $T/2$); while $ U_M(t)$ describes a competition where a team's utility is whether it wins more than its opponent. Notice that, when $U=U_E$ or $U=U_M$, we have $U(t)+U(T-t)=0$, hence both utility functions define a zero-sum game.
In this paper we always assume that $U(T)+U(T-t)=0$.

\subsection{Example: simultaneous card games}\label{subsect:card}

The models above formulates the standard team competitions as commonly seen under the context of sports, but shall not be limited to sports. The following is an instance of card games that fall into our framework.

Suppose that Alice and Bob each has a deck of three cards. In each deck one card is in suit~$\heartsuit$ and two cards are in suit~$\spadesuit$.
They play three rounds; in each round Alice and Bob select one card and they reveal the cards simultaneously.
If they select cards in same suit (both in $\heartsuit$ or both in $\spadesuit$), Alice wins this round; otherwise Bob wins this round.
The one who wins two or three rounds gets utility 1; the other one wins zero or one rounds and it gets utility -1.

This game can be conveniently represented in our model using the following parameters:
\begin{displaymath}
m=n=T=3, P=\left(
                    \begin{array}{ccc}
                      1 & 0 & 0 \\
                      0 & 1 & 1 \\
                      0 & 1 & 1 \\
                    \end{array}
                  \right),U=U_M.
\end{displaymath}
For this game, according to our first theorem, a (subgame perfect) equilibrium strategy for both players exists,
and it is to just to play uniformly random. It follows that Alice has expected utility $-1/3$ and Bob has $+1/3$.

\subsection{Extensive-form game with perfect recall}

Any particular instance $G(T,P,U)$ of our team competition is an extensive-form game.
In this game, a history can be described by a tuple $(k,\textbf{a},\textbf{b},\textbf{c})$ where:
\begin{itemize}
\item  $k$ indicates the number of rounds that has been played;
\item  $\textbf{a}$ is a $k$-dimensional vector which stores the players selected by Team~1 in the past $k$ rounds;
\item  $\textbf{b}$ stores the players selected by Team~2;
\item  $\textbf{c}$ is a $k$-dimensional $0-1$ vector which stores the results in the first $k$ rounds, where $0$ corresponds to a lose by Team~1 and $1$ corresponds to a win by Team~1.
  \end{itemize}
A behavioral strategy in this game is a mapping from every history to a probability distribution over actions.
That is, at each history, the strategy of each team is to pick the next player according to a probability distribution.
By Kuhn' Theorem (\cite{Course99,KuhnTHM}), there is a subgame perfect equilibrium (SPE), in which both teams use behavioral strategies.
According to the SPE, a value $V(H)$ can be defined for each history $H$, which indicates the expected utility that Team~1 would get at the end of the game if it is now at history $H$.
Note that each history is the root of a subgame and so the value of a history is the same as the value of the subgame.

\subsection{Computing SPE}\label{subsect:dp}

By backward induction, one can easily get
\begin{lemma}
Two histories have the same value if they have selected the same players to play (but may be in different orders) and Team~1 won the same number of rounds.
\end{lemma}

Based on the above lemma, the histories can be partitioned to equivalence classes, such that
each equivalence class corresponds to a four-tuple $(k,X,Y,w)$: $k$ is a number in $[T]$ which denotes the number of past rounds;
$X$ is a subset of $A$ of size $k$; $Y$ is a subset of $B$ of size $k$; $X,Y$ denote the players that have played;
$w$ is a number in $[k]$ which denotes how many rounds Team~1 has won so far.

\medskip In the following, we show in detail how to compute the value of each equivalence class via dynamic programming.

Let $V[k,X,Y,w]$ denote the expected utility of Team~1 when the history belongs to class $(k,X,Y,w)$.

Clearly, we have $V[k,X,Y,w]=U(w)$ when $k=T$.

When $k<T$, computing $V[k,X,Y,w]$ reduces to computing the value of the matrix game $M(k,X,Y,w)$ where the matrix $M(k,X,Y,w)$ is defined as follows.

It consists of $m-k$ rows and $n-k$ columns. Each row corresponds to a player in $A-X$, and each column corresponds to a player in $B-Y$.
The cell corresponding to $A_i,B_j$ equals to the expect utility of Team~1 when Team~1 and Team~2 respectively make action $A_i$ and $B_j$ on the current state $(k,X,Y,w)$, which equals
\begin{displaymath}
\begin{gathered}
V[k+1,X+\{A_i\}, Y+\{B_j\}, w + 1] \cdot  P_{i,j} +\\
V[k+1,X+\{A_i\}, Y+\{B_j\}, w] \cdot  (1- P_{i,j}).
\end{gathered}
\end{displaymath}

The reason behind the above definition of $M(k,X,Y,w)$ is as follows.
If the two teams select $A_i,B_j$ in this round, Team~1 has probability $P_{i,j}$ to win this round and hence the history becomes $(k+1,X+\{A_i\},Y+\{B_j\},w+1)$;
besides, Team~1 has probability $(1-P_{i,j})$ to lose this round and hence the history becomes $(k+1,X+\{A_i\},Y+\{B_j\},w)$.

We can compute the value of all equivalent classes of histories according to the above induction.
In fact, by computing these values, we also find a subgame perfect behaviorial strategy for both players.
To see this, suppose that $(k,X,Y,w)$ is a non-terminal equivalent class of history.
On solving the matrix game $M[k,X,Y,w]$ we find the strategies for all the histories in the history class $(k,X,Y,w)$.

\section{Uniformly Random Strategies}

The next theorem states that, if there are no redundant players, uniformly random is an equilibrium strategy for the teams.
It holds for arbitrary utility function including $U_E$ and $U_M$.

\begin{definition}The \emph{uniformly random strategy} is a behavioral strategy, in which a team always selects from the remaining players uniformly at random in each round.
\end{definition}

\begin{theorem}\label{thm:random}
When both teams have no redundant players (i.e. $n=m=T$), then it is a SPE when both teams apply the uniformly random strategy.
\footnote{Note that there could be other SPEs. For example, when players in Team~1 always lose, any strategies for the two teams form a SPE.}
\end{theorem}

We apply the following lemma for proving Theorem~\ref{thm:random}.

\begin{lemma}\label{lemma:random}
Suppose that there are no redundant players. Let $\mathbb{S}$ denote the set of all matchings between $\{A_1,\ldots,A_T\}$ and $\{B_1,\ldots,B_T\}$.
If Team~1 \textbf{or} Team~2 applies the uniformly random strategy, then the probability that the competition ends with any fixed matching in $\mathbb{S}$ is exactly $1/(T!)$.
\end{lemma}

\begin{proof}[of Lemma~\ref{lemma:random}]
First, suppose that Team~1 applies the uniformly random strategy while Team~2 applies an arbitrary pure strategy.
In this case, we claim that the probability that competition ends with any fixed matching is exactly $1/(T!)$.
This can be proved by induction on the number of players as the following.
Assume that Team~2 selects $B_i$ in the first stage, then the probability that it meets $A_j$ is $1/T$ for any $j$.
By induction hypothesis, any matching has equal probability to appear.
Then, since that any mixed strategy is a linear combination of the pure strategies, the above claim implies the statement given in Lemma~2.
\end{proof}

\begin{proof}[of Theorem~\ref{thm:random}] Let $\mathbb{S}$ denote the same set as in Lemma~\ref{lemma:random}.
Since there are no redundant players, a game will always end with some matching in $\mathbb{S}$.
For a matching $s\in \mathbb{S}$, let $Z_s$ denote the event that the game ends with this matching.
According to the lemma, if Team~1 applies the uniformly random strategy, it will get expected utility
\begin{displaymath}
  \sum_{s\in \mathbb{S}} E(\text{the utility of Team~1}\mid Z_s) / (T!)
\end{displaymath}
Similarly, if Team~2 applies the uniformly random strategy, it will get
\begin{displaymath}
  \begin{gathered}
    \sum_{s\in \mathbb{S}} E(\text{the utility of Team~2}\mid Z_s) / (T!)\\
    =\sum_{s\in \mathbb{S}} -E(\text{the utility of Team~1}\mid Z_s) / (T!)
  \end{gathered}
\end{displaymath}
Therefore, it is a Nash Equilibrium if both teams apply the uniformly random strategy.
The argument can be similarly extended to show that it is SPE.
\end{proof}

In the remainder of this paper, we focus on the case where there are redundant players.

\begin{claim} If there are redundant players, then the uniformly random strategy may \emph{not} be a SPE strategy.
\end{claim}

This claim is obvious for a team with redundant players; but less obvious for a team without redundant players.
Here we give an example in which the uniformly random strategy is not a SPE strategy for a team with no redundant players.

\begin{example}
Let $m=T=2,n=3$, $U=U_E=U_M$ (note that $U_E=U_M$ when $T=2)$, $P=\left(
         \begin{array}{ccc}
           0 & 0 & 1 \\
           1 & 1 & 0 \\
         \end{array}
       \right)$.
\end{example}

According to method given in Subsection~\ref{subsect:dp}, we can compute that $M(0,\varnothing,\varnothing,0)=\left(
         \begin{array}{ccc}
           -1 & -1 & 1 \\
           0 & 0 & -1 \\
         \end{array}
       \right)$.
Therefore, in the SPE, the behavior for Team~1 on the initial state
should be select $A_1$ with probability $1/3$ and select $A_2$ with probability $2/3$.
This guarantees (expected) utility $-1/3$.
If to the contrary that Team~1 adopts the uniformly random strategy, it should select $A_1,A_2$ with probability $1/2$,
which would only guarantees (expected) utility $-1/2$.

\subsection{Transitive strength}

\begin{definition}
Player $A_i$ is \emph{weaker} than its teammate $A_j$, denoted by $A_i\leq A_j$, if for any opponent $B_k$, the probability of ``$A_i$ wins against $B_k$'' is less than or equal to the probability of ``$A_j$ wins against $B_k$''. Similar for Team~2.

Team 1 $\{A_1,\ldots,A_m\}$ are \emph{transitive} if there is a permutation $\pi$ of $1,\ldots,m$, such that $A_{\pi(1)} \leq \ldots \leq  A_{\pi(m)}$. Similar for Team~2.
\end{definition}	

\begin{definition}
A utility function $U$ is \emph{monotone} if $U(t+1)\geq U(t)$ for $t\in 0,\ldots\,T-1$.
\end{definition}

\begin{theorem}\label{thm:transitive}
Assume the utility function  is monotone. We have that, (1) If $A_m \leq \ldots \leq A_1$, Team~1 has a SPE strategy which only selects $A_1,\ldots,A_T$. (2) Symmetrically, if $B_n \leq \ldots \leq B_1$, Team~2 has a SPE strategy which only selects $B_1,\ldots,B_T$.
\end{theorem}

By combining Theorem~\ref{thm:random} and Theorem~\ref{thm:transitive}, we can immediately get the following

\begin{corollary}\label{col:transitive}
When players in each team are transitive and $U$ is monotone, there is a simple SPE strategy for both teams as follows.
Assume that $A_m \leq \ldots \leq A_1$ and $B_n \leq \ldots \leq B_1$.
Then, the SPE strategy for Team~1 is to select an unused player in $A_1,\ldots,A_T$ uniformly random in each round;
a SPE strategy for Team~2 is to select an unused player in $B_1,\ldots,B_T$ uniformly random in each round.
\end{corollary}

Theorem~\ref{thm:transitive} and Corollary~\ref{col:transitive} have many applications.
In the real word, the utility function is monotone and moreover, in many situations such as in board games or some sport games, it is indeed the case that the players are transitive.

\medskip We prove Theorem~\ref{thm:transitive}~(1) in the next; the claim (2) is symmetric.

We first provide two basic terminologies which are necessary for understanding the subsequent proof.
Suppose that $A_m \leq \ldots \leq A_1$ and that $A'$ is a subset of $A$ and $A'=(A_{i[1]},\ldots A_{i[|A'|]})$, where $i[1] < \ldots < i[|A'|]$.
Then, for any $0\leq C\leq |A'|$, the top $C$ players of $A'$ refers to $\{A_{i[1]},\ldots,A_{i[C]}\}$,
and the rank $C$ player of $A'$ refers to $A_{i[C]}$.

\begin{lemma}\label{lemma:transitive}Suppose that $U()$ is monotone.

\begin{enumerate}
\item Consider a pair of history classes $H_1=(k, X_1, Y, w_1)$ and $H_2=(k,X_2, Y, w_2)$. We claim that,
if the top $T-k$ players of $A-X_1$ and the top $T-k$ players of $A-X_2$ are the same and $w_1\geq w_2$, then $V(H_1)\geq V(H_2)$.

\item Let $H=(k,X,Y,w)$ be a non-terminal history class. Let $A_u$ be the rank $T-k$ player in $A-X$, and let $A_v$ be any player in $A-X$ that is not a top $T-k$ player.
Then, the row in $M(k,X,Y,w)$ that corresponds to $A_u$ dominates the row that corresponds to $A_v$.
It further implies that, there is an equilibrium strategy at history $H$ (for Team~1) which only selects the top $T-k$ unmatched players to play.
\end{enumerate}
\end{lemma}

\begin{proof}
We prove it by backward induction.
When $k=T$, Claim 1 holds according to the monotone property of $U()$; and Claim 2 naturally holds since it is a terminal history.

Now, we argue that, for $0\leq k<T$, if the lemma holds for $k+1$, it also holds for $k$.

First, we prove Claim 2. Let us compare the two rows corresponding to $A_u$ and $A_v$.
Let us fix a column, say the one corresponding to $B_r$.
The cell corresponding to $(A_u,B_r)$ is
	\begin{displaymath}
        \begin{split}
         M[u,r]=&\underbrace{V(k+1,X+\{A_u\},Y+\{B_r\},w+1)}_a \cdot P_{u,r} \\
                &+ \underbrace{V(k+1,X+\{A_u\},Y+\{B_r\},w)}_b \cdot (1-P_{u,r})
        \end{split}
    \end{displaymath}
The cell corresponding to $(A_v,B_r)$ is
	\begin{displaymath}
        \begin{split}
         M[v,r]=&\underbrace{V(k+1,X+\{A_v\},Y+\{B_r\},w+1)}_{a'} \cdot P_{v,r}\\
                &+ \underbrace{V(k+1,X+\{A_u\},Y+\{B_r\},w)}_{b'} \cdot (1-P_{v,r})
         \end{split}
    \end{displaymath}
Notice that the top $T-k-1$ players in $A-X-\{A_u\}$ and $A-X-\{A_v\}$ are the same.
So, from the induction hypothesis, $a'\geq a \geq a' \geq b\geq b'\geq b$, i.e. $a=a'\geq b=b'$.

Since that $A_u$ is the top $T-k$ player while $A_v$ is not, player
$A_v$ is weaker than $A_u$, which means that $P_{u,r}\geq P_{v,r}$.

Combining the above arguments, we get that
    \begin{displaymath}
        M[u,r]-M[v,r]=(a - b) \cdot (P_{u,r}-P_{v,r}) \geq 0.
    \end{displaymath}
Therefore, $M[u,r]\geq M[v,r]$, and thus Claim~2 holds.

\medskip Then, we prove Claim 1.
Let $M_1$ denote $M(k,X_1,Y,w_1)$ and $M_2$ denote $M(k,X_2,Y,w_2)$ for short.
Suppose that $A_u$ is a top $T-k$ player in $A-X_1$ (which is also a top $T-k$ player in $A-X_2$) and that $B_r$ is any player in $B-Y$.

We know
	\begin{displaymath}
        \begin{split}
         M_1[u,r]=&V(k+1,X_1+\{A_u\},Y+\{B_r\},w_1+1) \cdot P_{u,r}+\\
                &V(k+1,X_1+\{A_u\},Y+\{B_r\},w_1) \cdot (1-P_{u,r})\\
         M_2[u,r]=&V(k+1,X_2+\{A_u\},Y+\{B_r\},w_2+1) \cdot P_{u,r}+\\
                &V(k+1,X_2+\{A_u\},Y+\{B_r\},w_2) \cdot (1-P_{u,r})
         \end{split}
    \end{displaymath}
By induction hypothesis, it follows that $M_1[u,r]\geq M_2[u,r]$.

Now, let $\sigma$ denote the equilibrium strategy at $H_2$ that only selects the top $T-k$ unmatched players to play. (Such a strategy exists according to Claim 2.)
Note that $\sigma$ is also a legal strategy at $H_1$. Let $\mu(H_1,\sigma)$ and $\mu(H_2,\sigma)$ respectively denote the utility of Team~1 when it applies strategy $\sigma$ on $H_1$ and $H_2$.
Then,
    \begin{displaymath}
        \begin{aligned}
            \mu(H_1,\sigma) &= &\min_{r:B_r\in Y}\sum_{u:A_u\in A-X_1}\sigma(A_u)\cdot M_1(u,r),\\
            \mu(H_2,\sigma) &= &\min_{r:B_r\in Y}\sum_{u:A_u\in A-X_2}\sigma(A_u)\cdot M_2(u,r).
        \end{aligned}
    \end{displaymath}

From the inequality $M_1(u,r)\geq M_2(u,r)$, we get $\mu(H_1,\sigma)\geq \mu(H_2,\sigma)$.
Moreover, we also have $V(H_1)\geq \mu(H_1,\sigma)$ and $V(H_2)=\mu(H_2,\sigma)$ (the equality is since that $\sigma$ is the equilibrium strategy on $H_2$).
Together, $V(H_1)\geq V(H_2)$.
\end{proof}

Finally, Claim 2 of Lemma~\ref{lemma:transitive} implies Theorem~\ref{thm:transitive}~(1).

\section{Nontransitive strength}

\begin{definition}
A player is said to be \emph{weakest}, if it is weaker than all its teammates; and is said to be \emph{dominated},
if it has 0 probability to win against any player in the opponent team.
\end{definition}

Assume that the utility function is monotone.
In the previous section, we show that if there are redundant players in Team~1 and if these players are transitive, then there is a SPE strategy for Team~1 which does not select the weakest player.
In other words, Team~1 can abandon the weakest one without decreasing its utility.
In this section, we show that the transitivity is essential for this to hold. We start by the following claim.

\begin{claim}
Suppose that Team~1 has redundant players and one of them is the weakest.
If the players in Team~1 are not transitive, Team~1 might decrease its utility by abandoning the weakest player.
\end{claim}

This is somewhat counterintuitive; it might be intuitive that the weakest player has no chance to participate in any match since one
can always replace him by a better teammate and increase utility. Perhaps even more surprisingly, we have the following claim:

\begin{claim}\label{claim:burden_is_useful}
Suppose that Team~1 has redundant players and one of them is dominated by the other team (i.e., has no chance to win at all).
If the players in Team~1 are not transitive, Team~1 might decrease its utility by abandoning the dominated player.
\end{claim}

The above claims confirms that, the weakest player or even dominated player could help its team.

We would now like to state the organization of the remainder of the section. In Subsection~\ref{subsect:useful}, we give examples that verify Claim~\ref{claim:burden_is_useful}, and we briefly explain the reason why we need dominated players.
In Subsection~\ref{subsect:useless}, we identify a special case where the weakest player can be abandoned without changing the utility.
In Subsection~\ref{subsect:number_of_burdens}, we consider the optimal number of dominated players that we may need to achieve maximum utility.
In Subsection~\ref{subsect:powerfulness_of_burdens}, we discuss the limitations of the dominated players.

\subsection{Dominated teammates can be helpful}\label{subsect:useful}

Let $V(T,P,U)$ denote the value of game $G(T,P,U)$.
Let $P^\ast$ denote the sub-matrix of $P$ by deleting the last row
(thus $G(T,P^\ast,U)$ is the game where Team~1 has abandoned $A_m$).

\begin{example}\label{example:useful}
Let $n=m=3,T=2$, $U=U_E$ (recall that $U_E(t)=t-T/2$), and
let $P=\left(
               \begin{array}{ccc}
                 1 & 0 & 0 \\
                 0 & 1 & 0 \\
                 0 & 0 & 0 \\
               \end{array}
             \right)$.
\end{example}

In Example~\ref{example:useful}, there are redundant players and the players in each team are not transitive. Besides, $A_3$ is a dominated player.
We argue the follows: (I) If $A_3$ is abandoned, Team~2 can win both rounds and hence $V(T,P^\ast,U)=-1$.
(II) If $A_3$ is in the team, Team~2 cannot win both rounds with certainty and that means $V(T,P,U)>-1$.
Combining (I) and (II), we get $V(T,P,U)>V(T,P^\ast,U)$, which implies Claim~\ref{claim:burden_is_useful}.

\begin{proof}[of (i)] If $A_3$ is abandoned, Team~2 can play as follows.
It chooses $B_3$ to win the first round.
If $B_3$ defeated $A_1$, it chooses $B_1$ in the second round to beat $A_2$;
otherwise, it chooses $B_2$ in the second round to beat $A_1$.
\end{proof}

\begin{proof}[of (ii)] If Team~2 wants to win with certainty in both rounds, it must select $B_3$ to play the first round.
However, if Team~1 selects the dominated player $A_3$ to play the first round, Team~2 cannot win the second round with certainty anymore.
\end{proof}

From this example, we see why a dominated player might be helpful for its team. The reason behind is similar to the horse race story described at the beginning of~\cite{Tang10}.

In the next, we give one more example. It gives, to our best knowledge, the largest decrease of utility by abandoning a dominated player.

\begin{example}\label{example:largest}
Let $m=4,n=T=3$. Let $U=U_E$ or $U=U_M$.
Let $P=\left(
         \begin{array}{ccc}
           1 & 0 & 0 \\
           0 & 1 & 0 \\
           0 & 0 & 1 \\
           0 & 0 & 0 \\
         \end{array}
       \right)
$.
\end{example}

According the method shown in Subsection~\ref{subsect:dp}, we can compute that
\footnote{Note that the value of $G(T,P^\ast,U_M)$ and $G(T,P^\ast,U_M)$ can be simply computed according to Theorem~\ref{thm:random} since there are no redundant players in these games.}
\begin{displaymath}
    \begin{aligned}
        V(T,P,U_M)&=0; & V(T,P^\ast,U_M)=&-2/3;\\
        V(T,P,U_E)&=-1/2; & V(T,P^\ast,U_E)=&-1/2.
    \end{aligned}
\end{displaymath}

So, for the game $G(T,P,U_M)$, we will lose utility as much as 2/3 if we abandon the dominated player.

In the following we explicitly state a SPE strategy for Team~1. In the first round it selects the dominated player $A_4$. Without loss of generality, assume that it loses to $B_1$.
In the second round, it selects $A_2,A_3$ uniformly random. So, there is $1/2$ chance that Team~1 wins this round.
Furthermore, if Team~1 wins the second round (say $A_2$ beats $B_2$) it can also wins the next (let $A_3$ beat $B_3$) and thus gets utility 1.
By this strategy, there is $1/2$ chance to get utility 1 and $1/2$ chance to get utility $-1$, so the expected utility is $0$.

However, for the game $(T,P,U_E)$, we do not lose any utility by abandoning the dominated player.
This is not a coincidence. In fact, this example belongs to a special case where the weakest player can indeed be abandoned. We show this in the next theorem.

\subsection{A case where the weakest player can be abandoned}\label{subsect:useless}

We have seen that the weakest redundant player is useless when the players are transitive (as proved in Theorem~\ref{thm:transitive})
but might be useful when the players are not transitive (as shown in the previous subsection).
So, the next question is:

\begin{question}
If the players are not transitive, on what cases the weakest player can be abandoned?
\end{question}

We get the following result.

\begin{theorem}\label{thm:special}
Suppose that Team~1 has redundant players but Team~2 does not. So, $m>T$ and $n=T$.
Suppose that each player in $A_{T+1},\ldots,A_m$ is weaker than each player in $A_1..A_T$. Let $U=U_E$.
Then, Team~1 can abandon all the players in $A_{T+1},\ldots,A_m$ without losing its utility.
\end{theorem}

\noindent Before we give the proof, we make two remarks.

1. The condition $m>n=T$ is important for Theorem~\ref{thm:special}.
If both team got redundant players, then the claim in this theorem does not hold anymore.

2. From Example~\ref{example:largest}, we see that the claim in this theorem fails when $U=U_M$.
Therefore, as a comparison, by recruiting extra dominated players, a team can gain more utility when $U=U_M$, but
cannot when $U=U_E$. This may suggest that $U_E$ is more reasonable than $U_M$ in team competition.

\medskip We need the following lemma on proving Theorem~\ref{thm:special}.
It is a technical statement of probability theory. Proof omitted.

\begin{lemma}\label{lemma:probability}
Assume that $n=T$ and $A_i,B_j$ are any pair of players from the two teams.
Let $Q^{\sigma}_{i,j}$ denote the probability that $A_i$ meets $B_j$ in the game when
Team~2 applies the uniformly random strategy and Team~1 applies some strategy $\sigma$.
Then, \begin{equation}\label{eqn:1/T}
Q^{\sigma}_{i,j}\leq 1/T.
\end{equation}
\end{lemma}

\begin{proof}[of Theorem~\ref{thm:special}]
We call $A_{T+1},\ldots,A_m$ the weak players.
When the weak players are abandoned, there are $T$ remaining players for each team.
By Theorem~\ref{thm:random}, the uniformly random strategy is a SPE strategy for Team~2.
To prove Theorem~\ref{thm:special}, the key idea is to show that even if Team~1 is allowed to select the weak player, it will not gain more utility if Team~2 keep using the uniformly random strategy.
On the other hand, it is obvious that Team~2 can't gain more utility. Therefore, the value of game does not change when the weak players are allowed to play.

When Team~1 abandons its weak players, its maximum utility is
\begin{displaymath}
  U^*=(\sum_{j=1..T}\sum_{i=1..T}\frac{1}{T}P_{i,j})-\frac{T}{2}.
\end{displaymath}

Let $Q^\sigma(i,j)$ be defined as Lemma~\ref{lemma:probability}.
The utility of Team~1 when it applies strategy $\sigma$ against the uniformly random strategy of Team~2 is
\begin{displaymath}
  U^\sigma=(\sum_{j=1..T}\sum_{i=1..m} Q^\sigma_{i,j}P_{i,j})-\frac{T}{2}
\end{displaymath}

We need to prove that $U^\sigma \leq U^*$,
and it reduces to show that for any fixed $j$ in $1..T$,
\begin{equation}\label{eqn:compare}
  \sum_{i=1..m}P_{i,j}Q^\sigma_{i,j}\leq \sum_{i=1..T}\frac{1}{T} P_{i,j}
\end{equation}

To prove (\ref{eqn:compare}), consider the following optimization problem:
\begin{displaymath}
\left\{
\begin{array}{cc}
  \texttt{Variables:} & x=(x_1,\ldots, x_m)\\
  \texttt{Parameters:} & c=(c_1,\ldots,c_m) \\
  \texttt{Guarantee:} & c_i\geq c_i' (\forall (i,i') \text{ such that } i\leq T<i') \\
  \texttt{Constrant~1:} & 0\leq x_i\leq \frac{1}{T} (\forall 1\leq i\leq m)\\
  \texttt{Constrant~2:} & \sum_{i=1}^{m} x_i=1 \\
  \texttt{Objective:} & \max f(x)=\sum_{i=1}^{m}c_ix_i
\end{array}
\right.
\end{displaymath}

Clearly, $f(x)$ is maximized at $x^\ast$, where $x^\ast_i=\left\{\begin{array}{cc}
                                                              \frac{1}{T} & i\leq T \\
                                                              0 & i>T
                                                            \end{array}\right.$.

\smallskip Noticing the following facts, we see that inequality (\ref{eqn:compare}) is just an application of the above problem.
\begin{displaymath}
\begin{array}{cc}
  Q^{\sigma}_{i,j}\leq \frac{1}{T} & \text{(Applying Lemma~\ref{lemma:probability})} \\
  \sum_{i=1}^{m} Q^\sigma_{i,j}=1 & \text{(According to the definition)} \\
  \forall i\leq T<i', P_{i,j}\geq P_{i',j} & \text{(Since $A_{i'}$ is weaker than $A_i$)}
\end{array}
\end{displaymath}
\end{proof}

\subsection{Optimal number of dominated players}\label{subsect:number_of_burdens}

Here we study the power of dominated players in another direction.
As we see in Subsection~\ref{subsect:useful}, by abandoning a redundant dominated player, Team~1 may decrease its utility.
In other words, it states that Team~1 may increase its utility by recruiting more dominated players.
Note that the utility of Team~1 will not decrease by recruiting more dominated players.
However, it is unclear that the utility will strictly increase by doing so.
For example, it is clear that recruiting $T$ dominated players is the same as recruiting $T-1$ players ---
in any case, if any team uses $T$ dominated players in a competition, it gets the lowest utility!

So, a natural question is,

\begin{question}
In order to get the maximum utility, how many dominated players should we recruit at least?
Is it possible that we need as many as $\Theta(T)$ such players?
\end{question}

The theorem below answers this question.

\begin{theorem}\label{thm:number_burdens}~~~
\begin{enumerate}
\item Suppose $U=U_E$. Recruiting $T-1$ dominated players can be better than $T-2$, but recruiting $T$ dominated players is the same as $T-1$.
 So, to achieve optimal utility, one may require $T-1$ dominated players. This number is tight.
\item Suppose $U=U_M$. Recruiting $\lfloor T/2 \rfloor$ dominated players can be better than $\lfloor T/2 \rfloor-1$,
    but recruiting $\lfloor T/2 \rfloor+1$ dominated players cannot be better than $\lfloor T/2 \rfloor$.
    So, to achieve optimal utility, one may require $\lfloor T/2 \rfloor$ dominated players, and this number is tight as well.
\end{enumerate}
\end{theorem}

One direction in these claims are rather trivial; we should never use $T$ dominated players when $U=U_E$ or $\lfloor T/2 \rfloor+1$ players when $U=U_M$. To prove the other direction, we need to construct some examples in which recruiting $T-1$ (resp. $\lfloor T/2 \rfloor$) could be better than $T-2$ (resp. $\lfloor T/2 \rfloor-1$) when $U=U_E$ (resp. $U=U_M$). To construct such examples, an intuition is that we should make the current players in Team~1 as weak as possible. Our construction is as follows.

\begin{example}\label{example:E}
\begin{displaymath}
    \begin{gathered}
    T\geq 1, m=T,n=T+(T-1),\\
    U=U_E, P_{i,j}=\left\{\begin{array}{cc}
                                    1 & i=j \\
                                    0 & i\neq j
                                  \end{array}
    \right.
    \end{gathered}
\end{displaymath}
\end{example}

\begin{example}\label{example:M}
\begin{displaymath}
    \begin{gathered}
    T\geq 1, m=T,n=T+\lfloor T/2 \rfloor,\\
    U=U_M, P_{i,j}=\left\{\begin{array}{cc}
                                    1 & i=j \\
                                    0 & i\neq j
                                  \end{array}
    \right.
    \end{gathered}
\end{displaymath}
\end{example}

The following claims together prove Theorem~\ref{thm:number_burdens}.

\begin{enumerate}
\item[C1.] In Example~\ref{example:E}, if Team~1 only recruit $T-2$ dominated players, it can win no rounds and thus can get utility $-T/2$.
\item[C2.] In Example~\ref{example:E}, if Team~1 recruit $T-1$ dominated players, it can win a positive number of rounds in expected and thus gain utility more than $-T/2$.
\item[C3.] In Example~\ref{example:M}, if Team~1 only recruit $\lfloor T/2 \rfloor-1$ dominated players, it will always lose at least $\lfloor T/2 \rfloor+1$ rounds and thus can only get utility $-1$.
\item[C4.] In Example~\ref{example:M}, if Team~1 recruit $\lfloor T/2 \rfloor$ dominated players, it can sometimes win at least $\lceil T/2\rceil $ rounds and thus can gain utility more than $-1$.
\end{enumerate}

\begin{proof}[of C1] In this case Team~2 can win all the rounds by playing as follows.
In the first $T-1$ rounds, it selects the players $B_{T+1},\ldots,B_{2T-1}$ to play; and they all win.
Then, since Team~1 only has $T-2$ dominated players, at least one player in $A_1,\ldots,A_T$ has already played, denote it by $A_i$.
In the last round, Team~2 select $B_i$ and it definitely wins.
\end{proof}

\begin{proof}[of C3] In this case, by applying a strategy similar to C1, Team~2 can win all the first $\lfloor T/2 \rfloor+1$ rounds.\footnote{In this case Team~2 can actually win all the $T$ rounds.}
\end{proof}

\begin{proof}[of C2]
For convenience, we denote the $T-1$ dominated players by $A_{T+1},\ldots,A_{2T-1}$.
We argue that, if Team~1 applies the uniform random strategy (that is, select one unused player in $A_1,\ldots, A_{2T-1}$ uniformly random in each round),
then, Team~2 has no strategy to win all rounds all the time.
Suppose to the contrary that Team~2 can do it, it must select a player from $B_{T+1}\ldots,B_{2T-1}$ to play in the first round; otherwise there is a chance that it loses the first round.
Note that, since Team~1 apply the uniform random strategy, there is a chance that Team~1 select a dominated player in the first round.
If this happens, Team~2 must again select a player from $B_{T+1}\ldots,B_{2T-1}$ to play in the second round.
Once again, Team~1 might still select a dominated player in the second round.
By induction, there is chance that Team~1 select all the dominated players in the first $T-1$ rounds while Team~2 consumes all its $T-1$ invincible players in $B_{T+1}\ldots,B_{2T-1}$.
Then, Team~2 cannot win with certainty in the last round.
\end{proof}

The claim C4 is the most nontrivial. To prove it we first state the following lemma.

\newcommand{\uC}{\lceil C/2 \rceil}
\newcommand{\dC}{\lfloor C/2 \rfloor}

\begin{definition} For integers $a,b,C$ such that
\begin{equation}\label{eqn:abC}
C\geq 1, 0\leq a\leq \uC, 0\leq b\leq \dC,
\end{equation}
let $\Gamma^C_{a,b}$ denote the following instance of team competition:

\begin{displaymath}
    \begin{gathered}
        m=n=(C-a)+(\dC -b),\\
        T=C-a-b,
        	P_{i,j} = \left\{ \begin{array}{cc}
        1 &  i=j\leq C-a\\
        0 & otherwise
        \end{array}
        \right..
    \end{gathered}
\end{displaymath}

The utility is as follows.\footnote{Here the utility functions for two teams are not identical.
However, since it is still a zero-sum game, SPE strategies for the teams exists as before.
The requirement that the utility functions are identical is not necessary in our model.}
If Team~1 wins at least $\uC-a$ rounds, it gets utility $1$ and Team~2 gets $-1$;
otherwise, Team~1 gets utility $-1$ and Team~2 gets $1$.
\end{definition}

\begin{lemma}\label{lemma:gamma_game}
For integers $a,b,C$ satisfying condition (\ref{eqn:abC}), Team~1 can win utility larger than $-1$ in the game $\Gamma^C_{a,b}$.
\end{lemma}

\begin{proof} Consider three cases.
\begin{itemize}
\item[Case~1] $a=\uC$.\\
    In this case, Team~1 always get utility $1$, and so $\Gamma^C_{a,b}$ has value $1$, which is larger than $-1$.
\item[Case~2] $b=\dC$.\\
	The game $\Gamma^C_{a,b}$ can be restated as follows.
    \begin{itemize}
    \item $m=n=C-a$, $T=\uC-a$. \\
        Player $A_i$ can only defeat $B_i$ for $i$ in $1..m$.\\
        Team~1 gets utility $1$ if it wins all the rounds; and $-1$ otherwise.
    \end{itemize}
	We argue that the uniformly random strategy guarantees Team~1 an expected utility larger than $-1$.
    Equivalently speaking, by applying the uniformly random strategy, Team~1 has a chance to win all the rounds.
	The proof is as follows. In the first round, there is a positive chance that $A_i$ meets $B_i$ for some $i$.
    Then, in the second round, same thing happens with a positive chance.
    This could happen for each round. When these coincidences happen, Team~1 wins all the rounds.
\item[Case~3] $a<\uC$ and $b<\dC$.\\
    We use induction. Assume that $\Gamma^C_{a+1,b}$ and $\Gamma^C_{a,b+1}$ both have value larger than $-1$, we argue that so does $\Gamma^C_{a,b}$.

    The following facts follow from the definition of $\Gamma^C_{a,b}$.
    \begin{description}
    \item[Fact~1.] If the two teams select $A_i$ and $B_i$ for $i\leq C-a$ in the first round, it becomes a subgame that is equivalent to $\Gamma^C_{a+1,b}$.
    \item[Fact~2.] If the two teams select $A_i$ and $B_i$ for $i>C-a$ in the first round, it becomes a subgame that is equivalent to $\Gamma^C_{a,b+1}$.
    \end{description}
    Combining them with the induction hypothesis, we get
    \begin{description}
    \item[Fact~3.] If the two teams select players under the same index, it becomes a subgame whose value is larger than $-1$.
    \end{description}
    We know that the value of $\Gamma^C_{a,b}$ is equal to the value of the matrix game $M$, where
    $M(i,j)$ indicate the value of the subgame when Team~1 select $A_i$ and Team~2 select $B_j$ in the first round.
	Fact~3 implies that all the utilities on the diagonal of matrix $M$ are larger than $-1$.
	Therefore, by using uniformly random strategy over its players, Team~1 can win a utility larger than $-1$.
    Therefore, $\Gamma^C_{a,b}$ has value larger than $-1$. 		
\end{itemize}
\end{proof}

\begin{proof}[of C4]
Let $G$ denote the revised game of Example~\ref{example:M}, in which Team~1 has recruited $\lfloor T/2 \rfloor$ dominated players.
We could observe that game $G$ is almost the same as $\Gamma^T_{0,0}$.
To be more specific, when $T$ is odd, $G$ is exactly $\Gamma^T_{0,0}$;
when $T$ is even, the parameters $m,n,T,P$ in $G$ and $\Gamma^T_{0,0}$ are the same;
but the utility $U$ is slightly different.

Suppose that the value of $G$ is $-1$. Then, Team~2 has a strategy which guarantees a expected utility $-1$.
It means that Team~2 has a strategy which can always win $\lfloor \frac{T}{2} \rfloor+1$ or more rounds.
When Team~2 applies this strategy, Team~1 can never win $\lceil \frac{T}{2} \rceil$ rounds.
It further implies that the value of $\Gamma^T_{0,0}$ is also $-1$.
However, this contradicts with Lemma~\ref{lemma:gamma_game}. Therefore, the value of $G$ must be larger than $-1$.
\end{proof}

\subsection{Limitations of the dominated players}\label{subsect:powerfulness_of_burdens}

Although the presences of dominated players can affect the value of the game,
we conjecture that it will not be too much. A question is then,

\begin{question}
By abandoning a dominated player, how much utility might be lost in the worst case?
In other words, how much can a team gain by recruiting dominated players?
\end{question}

According to our simulations, we have the following conjecture that we cannot prove at the moment.

\begin{conjecture}
If $U=U_M$, we can gain at most 2/3 extra utility by recruiting arbitrary number of dominated players.
If $U=U_E$, we can gain at most 1 extra utility by recruiting arbitrary number of dominated players.
\end{conjecture}

\section{Throwing a match and discarding a player}

Recall the card game between Alice and Bob in Subsection~\ref{subsect:card}.
We shall point out that, recruiting a dominated player in this context can be thought of applying a cheating action, which is to \emph{throw a match} by not placing any card in that round.
In the mentioned card game, if Alice and Bob are not allowed to throw a match, Alice can get expected utility $-1/3$;
if Alice is allowed to throw a match, she can get expected utility $1/3$.
This can be computed according to the method shown in Subsection~\ref{subsect:dp}.
Therefore, throwing a match is profitable if permitted.

It may seem unnatural to let a team throw a match like this. The following alternative cheating action called \emph{discarding}, which may seem more natural, is still profitable for the team.

Discarding is defined as follows. Alice (Team~1) is allowed to discard one of its cards and agrees to lose in that round; however the discarded card is never revealed to Bob (Team~2).
By discarding, Alice do not gain one more card at hand, unlike the case of throwing a match.

However, if discarding is allowed, it may still be beneficial. We give an instance in which one may gain extra utility by discarding.
Formally, we have the following theorem.

\begin{theorem}\label{thm:discarding}
There exists a game $G$ such that $V_K(G)>V_{K-1}(G)$, where $V_K(G)$ denotes the value of game $G$ in which Team~1 is allowed to discard at most $K$ players.
\end{theorem}

\begin{example}\label{example:D}
\begin{displaymath}
    \begin{gathered}
        m=K+1,n=2K+1,T=K+1,\\
        U=U_E, P=\left\{\begin{array}{cc}
                          1 & i=j \\
                          0 & i\neq j
                        \end{array}
        \right.
    \end{gathered}
\end{displaymath}
\end{example}

The following claims together imply Theorem~\ref{thm:discarding}.

\begin{itemize}
\item[C5] In Example~\ref{example:D}, if Team~1 is only allowed to discard $K-1$ times, it cannot win in any round.
\item[C6] In Example~\ref{example:D}, if Team~1 is allowed to discard $K$ times, it can win some rounds in expectation.
\end{itemize}

\begin{proof}[of C5] Suppose that Team~1 is only allowed to use discarding $K-1$ times.
Observer that, for any $i$ in $1..m$, after $A_i$ has played and revealed by Team~1, player $B_i$ becomes invincible that he would win with certainty if he plays in the next rounds.
Notice that there are $K$ invincible players at beginning (which are $B_{K+1}\ldots B_{2k+1}$) and Team~1 has only $K-1$ chances to hide a player by discarding.
So, in any round, Team~2 has an invincible player at hand. Therefore, Team~2 can win all the rounds.
\end{proof}

\begin{proof}[of C6] Consider the following strategy for Team~1. First, Team~1 randomly chooses an order of the players (say, each order with possibility $1/m!$), and then it randomly chooses exactly $K$ of its players so that these players will be discarded while playing. We argue that this strategy guarantees Team~1 to win positive rounds in expected. It reduces to prove that no strategy of Team~2 can win all the rounds against this strategy. Suppose that Team~2 can do so. In the first round, it must select an invincible player (i.e. a player in $B_{T+1}\ldots B_{T+k}$). Otherwise, there is a chance that it loses the first round. And then, we know that there is a chance Team~1 discards in the first round. If this happens, Team~2 must also select an invincible player in the second round. Again, it is possible that Team~1 still discards in this round. Continuing this process we see that, it could happen that, in the first $K$ rounds Team~2 use all its $K$ invincible players, while on the other hand Team~1 uses discarding $K$ times. Then, Team~2 could lose in the $K+1$-th round. \end{proof}

Note that, when discarding is allowed, the value of the game cannot be computed easily, because the competition becomes an extensive-game with imperfect information.

\section{Conclusions}

In this paper, we study a novel game-theoretic model of situations where two teams make sequential decisions about which of a set of exhaustible actions to select in each round.
These actions can be interpreted as team members, cards in a hand, etc.
We present a simple SPE for the case where there are no redundant players or the strength of players is transitive.
For the other case, we exhibit evidence that the redundant dominated players cannot be easily discounted in their contribution to team performance, which may appear counterintuitive.
We investigate the power of the dominated players in three directions:
  1. When do they influence the value of the competition?
  2. If additional dominated players can be recruited, how many should be required to attain the maximum utility?
  3. How much utility might be lost at most if we abandon them?
We obtain several nontrivial results that fully or partially answer these questions.
We believe that our results are of particular interests to both designers and players of team competitions.

\clearpage

\bibliographystyle{abbrv}

\end{document}